%% file: Talents.v3.tex
\documentclass{amsart}

\usepackage{cite,graphicx,amsmath,amsfonts,amssymb,psfrag}
\usepackage[usenames]{color}
\usepackage[nolist]{acronym} 

\def\testH1overH0{\begin{array}{c}
    {\mathcal{H}}_1 \\
    \gtrless \\
    {\mathcal{H}}_0
  \end{array}}

  \def\testD1overD0{{\scriptsize \begin{array}{l}
    {\mathcal{D}}_1 \\
    \gtrless \\
    {\mathcal{D}}_0
  \end{array}}}

\newtheorem{theorem}{Theorem}

\def\N0{N_{\textrm{0}}}

\begin{document}


\title[How to use our talents]{How to use our talents based on Information Theory - or spending time wisely}
\author[M. Chiani, May 2010]{
{Marco Chiani} \\
\smallskip
{University of Bologna \\
marco.chiani@unibo.it}
%
}


\maketitle
\date{\today}



\begin{abstract}
We discuss the allocation of finite resources in the presence of a logarithmic diminishing return law, in analogy to some results from Information Theory. To exemplify the problem we assume that the proposed logarithmic law is applied to the problem of how to spend our time.
\end{abstract}


\section{Aptitude, time and results}

\noindent Resource allocation is a key problem in economy \cite{Das:80}. We present some considerations inspired by Information Theory \cite{Sha:49,CovTho:B91} showing how, in general, whenever there is a diminishing return law of logarithmic type, the optimal allocation of the resources follows a water-filling behavior. 
%
To exemplify the treatment we analyze the economic problem of time allocation. 

\noindent More precisely, the problem we want to address is the following: we have some resources (time, energy,...) which can be used for different possible activities, such as:

\bigskip

$1$: playing piano

$2$: horse racing

$3$: driving motorbikes

...

$N$: working on new theorems.

\bigskip

\noindent We have a limited amount of resources - which we will indicate generally as \emph{time}. The total time available is $t_{tot}$.\footnote{This is the net time, after the effort we dedicate to mandatory activities.}

\noindent We are more skilled for some activities (they require us less \emph{time} to give good results). We will then speak about the \emph{aptitude} for the different activities (figure \ref{figATT}).

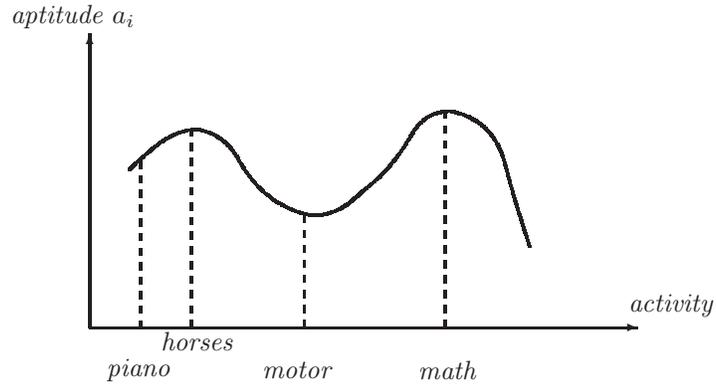
\begin{figure}[h!]
\centerline{\it \hspace{0cm} \input{figATTITUDINEcontscalENG.tex}} 
\caption{Aptitude.} \nonumber
\label{figATT}
\end{figure}

\noindent We assume that, as often in nature, the result is related to the time we spend on an activity, but in a less than proportional way. 
In other words, if we double the time we dedicate to an activity, the results we get will grow but less that twice (diminishing return). 
\noindent In particular, we assume that the results are described by a logarithmic rule (figure \ref{figLOGA}):\footnote{This logarithmic behavior is often encountered in Information Theory \cite{Sha:49}: for example, by spending a power $P$, the amount of information that can be sent through a channel impaired by additive Gaussian noise (mutual information between the input and the output) is proportional to $\log(1+P/N)$, where $N$ is the noise power.}
\begin{equation}\label{eq:risult}
 r_i = \log\left( 1 + a_i t_i \right) 
 \end{equation}
where, for the $i-th$ activity, $r_i$ represents the \emph{result} (the higher, the better), $a_i$ the \emph{aptitude} we have for that activity, and $t_i$ the \emph{time} we dedicate to it. Note that, since $\log(1)=0$, we get a result zero if we dedicate no time to an activity.

\begin{figure}[h]
\centerline{\it \hspace{0cm} \input{figLOGAscalENG.tex}} 
\caption{Result as a function of the spent time.} \nonumber
\label{figLOGA}
\end{figure}
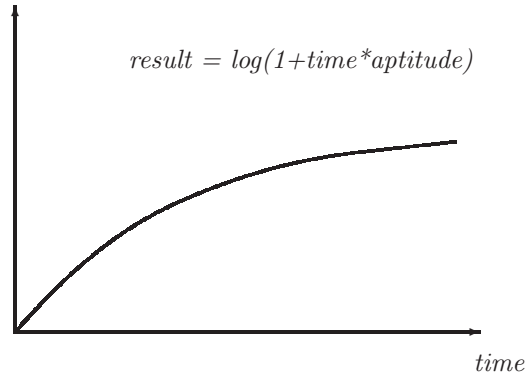

\noindent So, we have a set of {aptitudes} $a_1,\cdots, a_N$, and, if we allocate the times $t_1,\cdots,t_N$, we get the results  $r_1,\cdots,r_N$ given by  \eqref{eq:risult}. 

\noindent By spending the total time
\begin{equation}\label{eq:ttot}
t_{tot} = t_1 + t_2 + \cdots t_N
 \end{equation}
we get 
\begin{equation}\label{eq:ttot}
r_{tot} = r_1 + r_2 + \cdots r_N.
 \end{equation}
Now, the question is: how should we partition the total available time $t_{tot}$ to get the maximum of the overall result $r_{tot}$? Should we allocate more time to those activities where our skills are weaker? Or, should we spend more on those for which we have better aptitude?

\bigskip

\begin{theorem}(Resource allocation for logarithmic diminishing returns)
\label{th:wf}

In general, we must dedicate more time to the activities where we have the better aptitudes. Depending on the total available time, some of the activities where our aptitudes are worse must be abandoned. 

\bigskip

\noindent More precisely, to optimally allocate the time, there exist a minimum aptitude $a_{\min}$ such that:
\begin{itemize}
\item  all activities with $a_i<a_{\min}$ must be abandoned, i.e.:
\begin{equation}\label{eq:ti0}
t_{i} = 0 \qquad \text{if}  \qquad a_i \leq a_{\min}
 \end{equation}
\item for the activities with $a_i > a_{\min}$ we must allocate a time which is increasing with the aptitude, according to the rule
\begin{equation}\label{eq:ti}
t_{i} = 
\frac{1}{a_{\min}}-\frac{1}{a_i} >0. 
 \end{equation}
\end{itemize}
The threshold $a_{\min}$ is the value that, used in  \eqref{eq:ti0} and \eqref{eq:ti}, gives
\begin{equation*}\label{eq:amint}
t_1+t_2+\cdots +t_N = t_{tot}.
 \end{equation*}
 \end{theorem}

\begin{proof} (Analogous to \cite{Sha:49})
\noindent The problem consists in finding the $t_i \geq 0$ such that the overall result is maximized, under the contraint $t_1+\cdots + t_N=t_{tot}$:
\begin{eqnarray}
&& {t_1,...,t_N}=\arg\max_{t_1,...,t_N} \sum_{i=1}^{N} \log\left( 1 + a_i t_i \right)  \\
&& \nonumber \\
&with& \sum_{i=1}^{N} t_i = t_{tot} \label{eq:vincolottot}\\
& and & t_i \geq 0 .
\end{eqnarray}
%
 
\noindent By using the Lagrange's multipliers method, we set
$$
\frac{d}{d t_k} \left( \sum_{i=1}^{N} \log\left( 1 + a_i t_i\right)  - \lambda t_i  \right) = 0
$$
where $\lambda$ is the multiplier. Then, we have
$$
\frac{a_k}{1+ a_k t_k} -\lambda =0
$$
and therefore
\begin{eqnarray*}
t_k&=&\frac{1}{\lambda} -\frac{1}{a_k} \qquad \text{if greater than zero} \\
t_k&=&0  \qquad \text{if} \qquad  a_k < \lambda
\end{eqnarray*}
The value of $\lambda$ is that fulfilling \eqref{eq:vincolottot}. \footnote{We can use the Karush-Kuhn-Tucker conditions to verify that the proposed solution maximizes the global result.} In the Theorem we indicate $a_{\min}=\lambda$.
\end{proof}
  
\bigskip


\newpage

\section{Interpretation: water filling}

\noindent Similarly to \cite{Sha:49} for power allocation, the Theorem has a simple hydraulic interpretation.

\bigskip


\noindent In figure \ref{figINATT} we report the \emph{inaptitude} ($1/a_i$) for the different activities. We obtain a container with an irregular bottom, where the deepest the bottom, the largest the aptitude.

\begin{figure}[h!]
\centerline{\it \hspace{0cm} \input{figINATTITUDINEcontscalENG.tex}} 
\caption{Inaptitude.} \nonumber
\label{figINATT}
\end{figure}
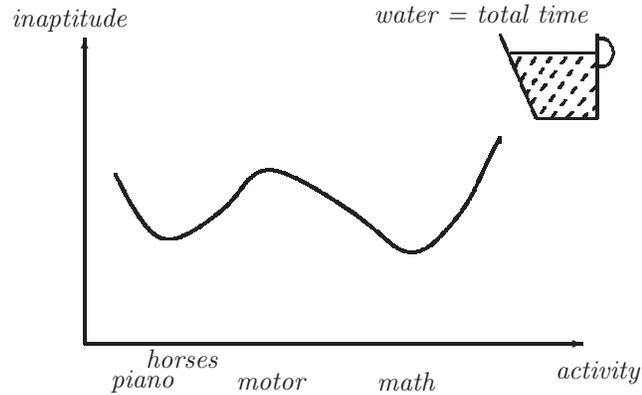
 
\noindent The total available time is represented by a certain amount of water. 

\noindent If we pour the water in the container, the depth of the water indicates the time to be dedicated to the activities.

\noindent The activities for which the inaptitude emerges above the water level must be abandoned. For the example of figure \ref{figTEMPOcont} it result that we should dedicate more time to horse-racing and math, while motorbike must be abandoned.

\begin{figure}[h!]
\psfrag{piano}{ppp}
\centerline{\it \hspace{0cm} \input{figTEMPOcontscalENG.tex}} 
\caption{Optimal time partition by water-filling. The total quantity of water is the available total time. The depth of the water indicates the time to dedicate to the activities. 
} \nonumber
\label{figTEMPOcont}
\end{figure}
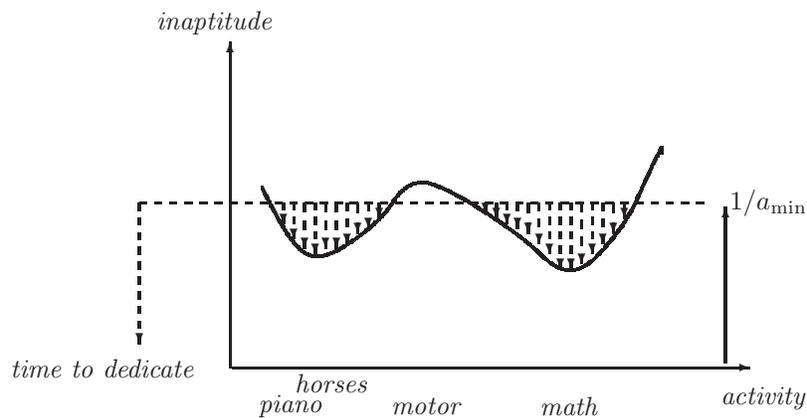

\bigskip

\noindent Therefore, if we have a lot of time we can cover all activities (but still the higher the aptitude for an activity, the larger the time we should dedicate to it).

\noindent On the other extreme, if we have a small amount of total time we should invest it on the activities where we have the best aptitude. This is the case illustrated in figure  \ref{figTEMPOcontPOCO}.

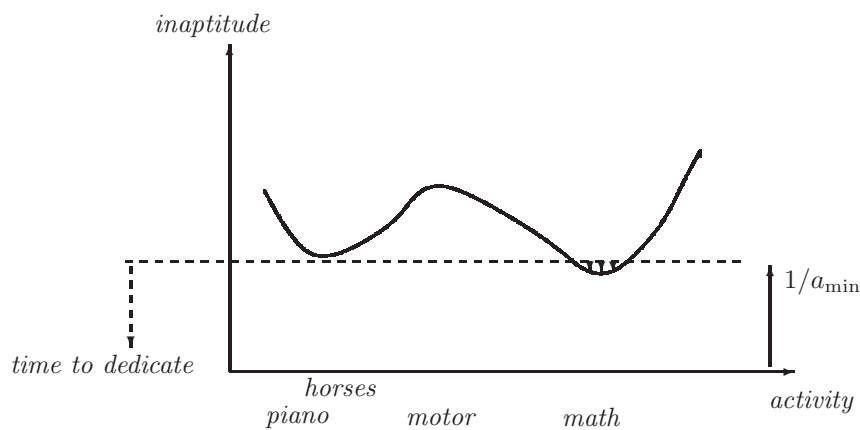
\begin{figure}[h!]
\centerline{\it \hspace{0cm} \input{figTEMPOcontPOCOscalENG.tex}} 
\caption{If the amount of time is small, we should better spend it on the activity where we have the best skills (in this case to math).} \nonumber
\label{figTEMPOcontPOCO}
\end{figure}

\bigskip

\section{Conclusions}

\noindent Assuming limited resources and a logarithmic relation between the product (aptitude*time) and the corresponding result, we should:
 
\begin{itemize}
\item abandon those activities for which the aptitude is below a certain threshold (the threshold depending on the total available time);
\item spend more time for the activities where we have the better aptitudes;
\item if the amount of available time is small, just dedicate it to the activity with the best aptitude.
\end{itemize}

%
%
%

\bibliographystyle{IEEEtran}

\end{document}

%% file: figATTITUDINEcontscalENG.tex
\ifx\JPicScale\undefined\def\JPicScale{1}\fi
\unitlength \JPicScale mm
\begin{picture}(105.75,51)(0,0)
\linethickness{0.3mm}
\put(28.5,9.6){\line(1,0){72.75}}
\put(101.25,9.6){\vector(1,0){0.12}}
\linethickness{0.3mm}
\put(28.5,9.6){\line(0,1){39}}
\put(28.5,48.6){\vector(0,1){0.12}}
\linethickness{0.3mm}
\put(105.75,12.6){\makebox(0,0)[cc]{activity}}

\put(26.25,51){\makebox(0,0)[cc]{aptitude $a_i$}}

\linethickness{0.3mm}
\qbezier(33.75,30.6)(33.96,30.97)(36.85,33.42)
\qbezier(36.85,33.42)(39.73,35.87)(42.75,36)
\qbezier(42.75,36)(46.47,35.34)(48.3,31.97)
\qbezier(48.3,31.97)(50.13,28.6)(53.25,26.4)
\qbezier(53.25,26.4)(54.61,25.6)(56.12,25.02)
\qbezier(56.12,25.02)(57.63,24.44)(59.25,24.6)
\qbezier(59.25,24.6)(61.61,25)(63.31,26.47)
\qbezier(63.31,26.47)(65.02,27.93)(66.75,29.4)
\qbezier(66.75,29.4)(69.34,31.88)(71.22,35.17)
\qbezier(71.22,35.17)(73.1,38.46)(76.5,38.4)
\qbezier(76.5,38.4)(82.13,37.3)(83.67,31.49)
\qbezier(83.67,31.49)(85.21,25.68)(87,20.4)
\put(35,4){\makebox(0,0)[cc]{piano}}

\put(42.75,7.8){\makebox(0,0)[cc]{horses}}

\put(56,4){\makebox(0,0)[cc]{motor}}

\put(76,4){\makebox(0,0)[cc]{math}}

\linethickness{0.3mm}
\multiput(35.25,9.6)(0,1.93){12}{\line(0,1){0.97}}
\linethickness{0.3mm}
\linethickness{0.3mm}
\multiput(42,9.6)(0,1.96){14}{\line(0,1){0.98}}
\linethickness{0.3mm}
\multiput(57,9.6)(0,2){8}{\line(0,1){1}}
\linethickness{0.3mm}
\multiput(75.75,9.6)(0,1.99){15}{\line(0,1){0.99}}
\end{picture}

%% file: figLOGAscalENG.tex
\ifx\JPicScale\undefined\def\JPicScale{1}\fi
\unitlength \JPicScale mm
\begin{picture}(89.29,53.33)(0,0)
\linethickness{0.3mm}
\put(25,10){\line(0,1){43.33}}
\put(25,53.33){\vector(0,1){0.12}}
\linethickness{0.3mm}
\put(89.29,6){\makebox(0,0)[cc]{time}}

\put(63,46){\makebox(0,0)[cc]{result = log(1+time*aptitude)}}

\linethickness{0.3mm}
\qbezier(25,10)(30.67,16.59)(36.08,20.85)
\qbezier(36.08,20.85)(41.5,25.1)(47.5,27.67)
\qbezier(47.5,27.67)(53.5,30.27)(59.22,31.8)
\qbezier(59.22,31.8)(64.95,33.32)(71.29,34)
\qbezier(71.29,34)(77.66,34.7)(80.6,35.02)
\qbezier(80.6,35.02)(83.54,35.34)(83.5,35.33)
\linethickness{0.3mm}
\put(25,10){\line(1,0){61.71}}
\put(86.71,10){\vector(1,0){0.12}}
\end{picture}

%% file: figINATTITUDINEcontscalENG.tex
\ifx\JPicScale\undefined\def\JPicScale{1}\fi
\unitlength \JPicScale mm
\begin{picture}(95.66,54)(0,0)
\linethickness{0.3mm}
\put(25.32,10){\line(1,0){66.13}}
\put(91.45,10){\vector(1,0){0.12}}
\linethickness{0.3mm}
\put(25.32,10){\line(0,1){40.62}}
\put(25.32,50.62){\vector(0,1){0.12}}
\linethickness{0.3mm}
\put(93.5,6.25){\makebox(0,0)[cc]{activity}}

\put(23.27,53.12){\makebox(0,0)[cc]{inaptitude}}

\linethickness{0.3mm}
\qbezier(29.41,32.5)(31.18,28.91)(32.5,26.96)
\qbezier(32.5,26.96)(33.81,25)(34.86,24.37)
\qbezier(34.86,24.37)(35.92,23.71)(37.31,24.01)
\qbezier(37.31,24.01)(38.71,24.31)(40.66,25.62)
\qbezier(40.66,25.62)(42.61,26.92)(43.92,28.12)
\qbezier(43.92,28.12)(45.24,29.32)(46.11,30.62)
\qbezier(46.11,30.62)(47,31.92)(47.9,32.53)
\qbezier(47.9,32.53)(48.8,33.13)(49.86,33.12)
\qbezier(49.86,33.12)(50.92,33.13)(52.64,32.38)
\qbezier(52.64,32.38)(54.36,31.63)(57.02,30)
\qbezier(57.02,30)(59.68,28.38)(61.73,26.88)
\qbezier(61.73,26.88)(63.78,25.37)(65.54,23.75)
\qbezier(65.54,23.75)(67.31,22.11)(68.95,22.18)
\qbezier(68.95,22.18)(70.59,22.25)(72.36,24.06)
\qbezier(72.36,24.06)(74.13,25.84)(75.37,27.57)
\qbezier(75.37,27.57)(76.6,29.3)(77.48,31.25)
\qbezier(77.48,31.25)(78.36,33.2)(78.94,34.41)
\qbezier(78.94,34.41)(79.51,35.61)(79.87,36.25)
\qbezier(79.87,36.25)(80.22,36.9)(80.39,37.2)
\qbezier(80.39,37.2)(80.55,37.5)(80.55,37.5)
\qbezier(80.55,37.5)(80.55,37.5)(80.55,37.35)
\qbezier(80.55,37.35)(80.55,37.2)(80.55,36.88)
\put(33,5){\makebox(0,0)[cc]{piano}}

\put(38.27,8.12){\makebox(0,0)[cc]{horses}}

\linethickness{0.3mm}
\multiput(80.55,51.25)(0.12,-0.28){40}{\line(0,-1){0.28}}
\linethickness{0.3mm}
\put(85.32,40){\line(1,0){8.18}}
\linethickness{0.3mm}
\put(93.5,40){\line(0,1){11.25}}
\linethickness{0.2mm}
\put(81.91,48.75){\line(1,0){11.59}}
\linethickness{0.3mm}
\multiput(82.59,46.88)(1.36,2.48){1}{\multiput(0,0)(0.11,0.21){6}{\line(0,1){0.21}}}
\linethickness{0.3mm}
\multiput(83.27,45)(1.37,2.08){2}{\multiput(0,0)(0.11,0.17){6}{\line(0,1){0.17}}}
\linethickness{0.3mm}
\multiput(83.95,43.12)(0.97,1.43){4}{\multiput(0,0)(0.12,0.18){4}{\line(0,1){0.18}}}
\linethickness{0.3mm}
\multiput(84.64,41.88)(1.36,1.78){4}{\multiput(0,0)(0.11,0.15){6}{\line(0,1){0.15}}}
\linethickness{0.3mm}
\multiput(86,40.62)(1.21,1.67){5}{\multiput(0,0)(0.12,0.17){5}{\line(0,1){0.17}}}
\linethickness{0.3mm}
\multiput(87.36,40)(1.36,1.67){5}{\multiput(0,0)(0.11,0.14){6}{\line(0,1){0.14}}}
\linethickness{0.3mm}
\multiput(91.45,40)(1.37,1.25){2}{\multiput(0,0)(0.14,0.13){5}{\line(1,0){0.14}}}
\linethickness{0.3mm}
\multiput(89.41,40)(1.64,1.75){3}{\multiput(0,0)(0.12,0.13){7}{\line(0,1){0.13}}}
\put(78,54){\makebox(0,0)[cc]{water = total time}}

\put(50,5){\makebox(0,0)[cc]{motor}}

\linethickness{0.3mm}
\multiput(93.5,46.88)(0.5,-0.01){1}{\line(1,0){0.5}}
\multiput(94,46.87)(0.48,0.12){1}{\line(1,0){0.48}}
\multiput(94.48,47)(0.22,0.12){2}{\line(1,0){0.22}}
\multiput(94.91,47.24)(0.12,0.12){3}{\line(1,0){0.12}}
\multiput(95.26,47.59)(0.12,0.21){2}{\line(0,1){0.21}}
\multiput(95.51,48.02)(0.13,0.48){1}{\line(0,1){0.48}}
\put(95.64,48.5){\line(0,1){0.5}}
\multiput(95.51,49.48)(0.13,-0.48){1}{\line(0,-1){0.48}}
\multiput(95.26,49.91)(0.12,-0.21){2}{\line(0,-1){0.21}}
\multiput(94.91,50.26)(0.12,-0.12){3}{\line(1,0){0.12}}
\multiput(94.48,50.5)(0.22,-0.12){2}{\line(1,0){0.22}}
\multiput(94,50.63)(0.48,-0.12){1}{\line(1,0){0.48}}
\multiput(93.5,50.62)(0.5,0.01){1}{\line(1,0){0.5}}

\put(68,5){\makebox(0,0)[cc]{math}}

\end{picture}

%% file: figTEMPOcontscalENG.tex
\ifx\JPicScale\undefined\def\JPicScale{1}\fi
\unitlength \JPicScale mm
\begin{picture}(96.54,56)(0,0)
\linethickness{0.3mm}
\put(25.29,10){\line(1,0){68.71}}
\put(94,10){\vector(1,0){0.12}}
\linethickness{0.3mm}
\put(25,10){\line(0,1){43.33}}
\put(25,53.33){\vector(0,1){0.12}}
\linethickness{0.3mm}
\put(95.83,6){\makebox(0,0)[cc]{activity}}

\put(22.88,56){\makebox(0,0)[cc]{inaptitude}}

\linethickness{0.2mm}
\multiput(12.96,32)(2,0){38}{\line(1,0){1}}
\linethickness{0.3mm}
\multiput(42,28)(0,2.67){2}{\line(0,1){1.33}}
\put(42,28){\vector(0,-1){0.12}}
\put(96.54,32){\makebox(0,0)[cc]{$1/a_{\min}$}}

\linethickness{0.3mm}
\multiput(12.96,13.33)(0,1.96){10}{\line(0,1){0.98}}
\put(12.96,13.33){\vector(0,-1){0.12}}
\put(8,10){\makebox(0,0)[cc]{time to dedicate}}

\linethickness{0.3mm}
\put(90.88,10.67){\line(0,1){20.67}}
\put(90.88,31.33){\vector(0,1){0.12}}
\linethickness{0.3mm}
\qbezier(29.25,34)(31.09,30.18)(32.46,28.09)
\qbezier(32.46,28.09)(33.82,26.01)(34.92,25.33)
\qbezier(34.92,25.33)(36.01,24.63)(37.46,24.95)
\qbezier(37.46,24.95)(38.91,25.27)(40.94,26.67)
\qbezier(40.94,26.67)(42.97,28.05)(44.33,29.33)
\qbezier(44.33,29.33)(45.69,30.62)(46.6,32)
\qbezier(46.6,32)(47.52,33.39)(48.46,34.03)
\qbezier(48.46,34.03)(49.4,34.67)(50.5,34.67)
\qbezier(50.5,34.67)(51.59,34.68)(53.38,33.87)
\qbezier(53.38,33.87)(55.17,33.07)(57.94,31.33)
\qbezier(57.94,31.33)(60.7,29.6)(62.83,28)
\qbezier(62.83,28)(64.96,26.4)(66.79,24.67)
\qbezier(66.79,24.67)(68.63,22.92)(70.33,23)
\qbezier(70.33,23)(72.04,23.08)(73.88,25)
\qbezier(73.88,25)(75.72,26.9)(77,28.75)
\qbezier(77,28.75)(78.27,30.59)(79.19,32.67)
\qbezier(79.19,32.67)(80.11,34.75)(80.71,36.03)
\qbezier(80.71,36.03)(81.3,37.32)(81.67,38)
\qbezier(81.67,38)(82.04,38.7)(82.21,39.02)
\qbezier(82.21,39.02)(82.38,39.34)(82.38,39.33)
\qbezier(82.38,39.33)(82.38,39.34)(82.38,39.18)
\qbezier(82.38,39.18)(82.38,39.01)(82.38,38.67)
\linethickness{0.3mm}
\multiput(32.08,29.33)(0,1.78){2}{\line(0,1){0.89}}
\put(32.08,29.33){\vector(0,-1){0.12}}
\linethickness{0.3mm}
\multiput(33.5,27.33)(0,1.87){3}{\line(0,1){0.93}}
\put(33.5,27.33){\vector(0,-1){0.12}}
\linethickness{0.3mm}
\multiput(34.92,26)(0,2.4){3}{\line(0,1){1.2}}
\put(34.92,26){\vector(0,-1){0.12}}
\linethickness{0.3mm}
\multiput(36.33,25.33)(0,1.9){4}{\line(0,1){0.95}}
\put(36.33,25.33){\vector(0,-1){0.12}}
\linethickness{0.3mm}
\multiput(37.75,26)(0,2.4){3}{\line(0,1){1.2}}
\put(37.75,26){\vector(0,-1){0.12}}
\linethickness{0.3mm}
\multiput(39.17,26)(0,2.4){3}{\line(0,1){1.2}}
\put(39.17,26){\vector(0,-1){0.12}}
\linethickness{0.3mm}
\multiput(40.58,26.67)(0,2.13){3}{\line(0,1){1.07}}
\put(40.58,26.67){\vector(0,-1){0.12}}
\linethickness{0.3mm}
\multiput(43.42,28.67)(0,2.22){2}{\line(0,1){1.11}}
\put(43.42,28.67){\vector(0,-1){0.12}}
\linethickness{0.3mm}
\multiput(44.83,30)(0,4){1}{\line(0,1){2}}
\put(44.83,30){\vector(0,-1){0.12}}
\linethickness{0.3mm}
\multiput(74.58,26.67)(0,2.13){3}{\line(0,1){1.07}}
\put(74.58,26.67){\vector(0,-1){0.12}}
\linethickness{0.3mm}
\multiput(64.67,28)(0,2.67){2}{\line(0,1){1.33}}
\put(64.67,28){\vector(0,-1){0.12}}
\linethickness{0.3mm}
\multiput(66.08,26.67)(0,2.13){3}{\line(0,1){1.07}}
\put(66.08,26.67){\vector(0,-1){0.12}}
\linethickness{0.3mm}
\multiput(67.5,25.33)(0,1.9){4}{\line(0,1){0.95}}
\put(67.5,25.33){\vector(0,-1){0.12}}
\linethickness{0.3mm}
\multiput(68.92,23.33)(0,1.93){5}{\line(0,1){0.96}}
\put(68.92,23.33){\vector(0,-1){0.12}}
\linethickness{0.3mm}
\multiput(70.33,23.33)(0,1.93){5}{\line(0,1){0.96}}
\put(70.33,23.33){\vector(0,-1){0.12}}
\linethickness{0.3mm}
\multiput(71.75,24)(0,2.29){4}{\line(0,1){1.14}}
\put(71.75,24){\vector(0,-1){0.12}}
\linethickness{0.3mm}
\multiput(73.17,25.33)(0,1.9){4}{\line(0,1){0.95}}
\put(73.17,25.33){\vector(0,-1){0.12}}
\linethickness{0.3mm}
\multiput(76,28)(0,2.67){2}{\line(0,1){1.33}}
\put(76,28){\vector(0,-1){0.12}}
\linethickness{0.3mm}
\multiput(77.42,30)(0,4){1}{\line(0,1){2}}
\put(77.42,30){\vector(0,-1){0.12}}
\linethickness{0.3mm}
\multiput(59,30.67)(0,2.67){1}{\line(0,1){1.33}}
\put(59,30.67){\vector(0,-1){0.12}}
\linethickness{0.3mm}
\multiput(60.42,30)(0,4){1}{\line(0,1){2}}
\put(60.42,30){\vector(0,-1){0.12}}
\linethickness{0.3mm}
\multiput(61.83,29.33)(0,1.78){2}{\line(0,1){0.89}}
\put(61.83,29.33){\vector(0,-1){0.12}}
\linethickness{0.3mm}
\multiput(63.25,28.67)(0,2.22){2}{\line(0,1){1.11}}
\put(63.25,28.67){\vector(0,-1){0.12}}
\put(33,5){\makebox(0,0)[cc]{piano}}

\put(38.46,8){\makebox(0,0)[cc]{horses}}

\put(51,5){\makebox(0,0)[cc]{motor}}

\put(70,5){\makebox(0,0)[cc]{math}}

\end{picture}

%% file: figTEMPOcontPOCOscalENG.tex
\ifx\JPicScale\undefined\def\JPicScale{1}\fi
\unitlength \JPicScale mm
\begin{picture}(103.82,56)(0,0)
\linethickness{0.3mm}
\put(25,10){\line(1,0){74.95}}
\put(99.95,10){\vector(1,0){0.12}}
\linethickness{0.3mm}
\put(25,10){\line(0,1){43.33}}
\put(25,53.33){\vector(0,1){0.12}}
\linethickness{0.3mm}
\put(102.27,6){\makebox(0,0)[cc]{activity}}

\put(22.68,56){\makebox(0,0)[cc]{inaptitude}}

\linethickness{0.2mm}
\multiput(11.09,24.67)(2.02,0){41}{\line(1,0){1.01}}
\put(103.82,22){\makebox(0,0)[cc]{$1/a_{\min}$}}

\linethickness{0.3mm}
\multiput(11.86,13.33)(0,1.94){6}{\line(0,1){0.97}}
\put(11.86,13.33){\vector(0,-1){0.12}}
\put(8,11){\makebox(0,0)[cc]{time to dedicate}}

\linethickness{0.3mm}
\put(96.86,10.67){\line(0,1){13.33}}
\put(96.86,24){\vector(0,1){0.12}}
\linethickness{0.3mm}
\qbezier(29.64,34)(31.65,30.18)(33.13,28.17)
\qbezier(33.13,28.17)(34.62,26.17)(35.82,25.67)
\qbezier(35.82,25.67)(37.01,25.14)(38.6,25.46)
\qbezier(38.6,25.46)(40.18,25.78)(42.39,27)
\qbezier(42.39,27)(44.6,28.21)(46.09,29.41)
\qbezier(46.09,29.41)(47.57,30.62)(48.57,32)
\qbezier(48.57,32)(49.57,33.39)(50.59,34.03)
\qbezier(50.59,34.03)(51.61,34.67)(52.82,34.67)
\qbezier(52.82,34.67)(54.01,34.68)(55.96,33.87)
\qbezier(55.96,33.87)(57.91,33.07)(60.93,31.33)
\qbezier(60.93,31.33)(63.94,29.6)(66.27,28)
\qbezier(66.27,28)(68.59,26.4)(70.59,24.67)
\qbezier(70.59,24.67)(72.6,22.92)(74.45,23)
\qbezier(74.45,23)(76.31,23.08)(78.32,25)
\qbezier(78.32,25)(80.33,26.9)(81.72,28.75)
\qbezier(81.72,28.75)(83.12,30.59)(84.11,32.67)
\qbezier(84.11,32.67)(85.12,34.75)(85.77,36.03)
\qbezier(85.77,36.03)(86.42,37.32)(86.82,38)
\qbezier(86.82,38)(87.22,38.7)(87.41,39.02)
\qbezier(87.41,39.02)(87.59,39.34)(87.59,39.33)
\qbezier(87.59,39.33)(87.59,39.34)(87.59,39.18)
\qbezier(87.59,39.18)(87.59,39.01)(87.59,38.67)
\linethickness{0.3mm}
\multiput(72.91,23.33)(0,2.67){1}{\line(0,1){1.33}}
\put(72.91,23.33){\vector(0,-1){0.12}}
\linethickness{0.3mm}
\multiput(74.45,23.33)(0,2.67){1}{\line(0,1){1.33}}
\put(74.45,23.33){\vector(0,-1){0.12}}
\linethickness{0.3mm}
\multiput(76,23.33)(0,2.67){1}{\line(0,1){1.33}}
\put(76,23.33){\vector(0,-1){0.12}}
\put(34,4){\makebox(0,0)[cc]{piano}}

\put(39.68,8){\makebox(0,0)[cc]{horses}}

\put(53,4){\makebox(0,0)[cc]{motor}}

\put(73,4){\makebox(0,0)[cc]{math}}

\end{picture}